%% file: main.tex
\documentclass{article}
\usepackage{graphicx} 

\usepackage{amsmath,amsthm,amssymb,amsfonts}
\usepackage{caption}
\usepackage{subcaption}
\usepackage{bbm}
\usepackage{graphicx}
\usepackage{multirow}
\usepackage{xspace}
\usepackage[colorlinks=true,citecolor=blue,linkcolor=blue]{hyperref}
\usepackage{cleveref}
\usepackage{siunitx,booktabs}
\usepackage{algorithmic}

\usepackage[letterpaper,top=2cm,bottom=2cm,left=3cm,right=3cm,marginparwidth=1.75cm]{geometry}

\usepackage{enumerate}

\usepackage{tikz}
\usetikzlibrary{arrows,shapes,automata, calc, chains, matrix, trees, positioning, scopes}
\usetikzlibrary{decorations.pathmorphing}
\tikzset{snake it/.style={decorate, decoration=snake}}

\newtheorem{theorem}{Theorem}
\newtheorem{proposition}[theorem]{Proposition}

\newtheorem{lemma}[theorem]{Lemma}

\newtheorem*{problem*}{Problem}

\numberwithin{theorem}{section}

\newcommand{\B}{\mathbb{B}}

\newcommand{\RNN}{\mathbb{R}_{\ge 0}}

\theoremstyle{definition}
\newtheorem{example}[theorem]{Example}

\newcommand{\comment}[1]{}

\newcommand{\ACZ}{\textnormal{\textsf{AC$^0$}}\xspace}

\newcommand{\NL}
{\textnormal{\textsf{NL}}\xspace}
\newcommand{\coNL}
{\textnormal{\textsf{coNL}}\xspace}
\newcommand{\LL}{\textnormal{\textsf{L}}\xspace}

\newcommand{\ZZ}{\mathbb{Z}}

\DeclareMathOperator{\modu}{mod}

\title{The complexity of computing \\ the period and the exponent of a digraph}
\author{Stefan Kiefer, Andrew Ryzhikov}
\date{Department of Computer Science, University of Oxford, UK}

\begin{document}

\maketitle

\begin{abstract}
    The period of a strongly connected digraph is the greatest common divisor of the lengths of all its cycles. The period of a digraph is the least common multiple of the periods of its strongly connected components. These notions play an important role in the theory of Markov chains and the analysis of powers of nonnegative matrices. While the time complexity of computing the period is well-understood, little is known about its space complexity. We show that the problem of computing the period of a digraph is \NL-complete, even if all its cycles are contained in the same strongly connected component. However, if the digraph is strongly connected, we show that this problem becomes \LL-complete. For primitive digraphs (that is, strongly connected digraphs of period one), there always exists a number $m$ such that there is a path of length exactly $m$ between every two vertices. We show that computing the smallest such $m$, called the exponent of a digraph, is \NL-complete. The exponent of a primitive digraph is a particular case of the index of convergence of a nonnegative matrix, which we also show to be computable in \NL, and thus \NL-complete.
\end{abstract}

\section{Introduction}\label{sec:intro}
\input{sec-intro}

\section{The period of a general digraph}\label{sec:general}
\input{sec-general}

\section{The period of a strongly connected digraph}\label{sec:sc}
\input{sec-sc}

\section{Exponent of primitivity of a digraph}\label{sec:exponent}
\input{sec-exponent}

\section{Conclusions}

We have shown that the period of a digraph can be computed in \NL, and distinguishing between periods one and two is already \NL-hard for almost strongly connected digraphs.
For strongly connected digraphs, the complexity drops to \LL-completeness.
We have shown that the index of convergence is computable in~\NL, complemented by an \NL-hardness result even for computing the exponent of primitive digraphs.
Thus, the computational complexity of the mentioned problems is very clear.
The literature contains linear-time algorithms for computing the period.
One might ask whether the exponent or index of convergence can similarly be computed in linear time.

\subsection*{Acknowledgements}
Andrew Ryzhikov is supported by the European Research Council (ERC) under the European Union’s Horizon 2020 research and innovation programme (Grant agreement No. 852769, ARiAT).

\bibliographystyle{alpha}
\bibliography{sync}

\end{document}

%% file: sec-intro.tex
A \emph{nonnegative matrix} is a matrix with nonnegative real entries. An $n \times n$ nonnegative matrix is called \emph{primitive} if there is $k \ge 0$ such that all entries of the $k$th power of the matrix under usual matrix multiplication are strictly positive.
Primitive matrices play an important role in combinatorial matrix theory, the analysis of Markov chains and symbolic dynamics. A classical result in the theory of Markov chains asserts that if the transition matrix of a finite Markov chain is primitive, this Markov chain has a unique stationary distribution~\cite{Seneta2006}. The road colouring theorem, an important result in symbolic dynamics and synchronising automata conjectured in 1970 \cite{Adler1970} and proved in 2009 \cite{Trahtman2009}, states that a digraph with the property that the outdegrees of all vertices are equal admits a colouring into a synchronising deterministic automaton if and only if its adjacency matrix is primitive.

To reason about more general patterns of zero entries in powers of a nonnegative matrix, the following natural homomorphism is useful. Let $\chi(A)$ be the matrix obtained from a nonnegative matrix $A$ by replacing every strictly positive entry with 1. 
Denote by $\RNN$ the set of nonnegative real numbers, and by $\B$ the Boolean semiring $\{0, 1\}$ (with usual addition and multiplication except that $1 + 1 = 1$). The mapping $\chi\colon \RNN^{n \times n} \to \B^{n \times n}$ is then a homomorphism from the monoid of nonnegative matrices to the monoid of matrices over $\B$; that is, $\chi(A \cdot B) = \chi(A) \cdot \chi(B)$ holds for all $A,B \in \RNN^{n \times n}$. Intuitively speaking, $\chi$ agrees with matrix multiplication.

Given a nonnegative matrix $A$, consider the sequence 
\begin{equation*}\label{eq-seq-a}
    \chi(A^1), \chi(A^2), \ldots\,,
\end{equation*}
where $A^k$ is the $k$th power of $A$ under usual matrix multiplication over the reals. Since $\chi$ is a monoid homomorphism, this sequence coincides with the sequence of increasing powers 
\begin{equation*}\label{eq-seq-chia}
    \chi(A)^1, \chi(A)^2, \ldots\,
\end{equation*}
of the element $\chi(A)$ from $\B^{n \times n}$. Since $\B^{n \times n}$ is a finite monoid under multiplication, this sequence is ultimately periodic: there exist natural numbers $m \ge 0$, $p \ge 1$ such that for all $i \ge m$ we have $\chi(A)^{i + p} = \chi(A)^i$. The~smallest such~$p$ is called the \emph{period} of~$A$, and the smallest such $m$ is called \emph{the index of convergence} of~$A$~\cite{Heap1966,Jia1987}.

\begin{example} \label{ex:1}
    Consider the nonnegative matrix
    \[
    A = \begin{pmatrix}
        1 & 1 & 0 \\
        0 & 0 & 1 \\
        0 & 1 & 0
    \end{pmatrix} \ \in \ \RNN^{3 \times 3}\,.
    \]
    We have
    \[
    A^2  =  \begin{pmatrix}
        1 & 1 & 1 \\
        0 & 1 & 0 \\
        0 & 0 & 1        
    \end{pmatrix} \,, \ 
    A^3  =  \begin{pmatrix}
        1 & 2 & 1 \\
        0 & 0 & 1 \\
        0 & 1 & 0        
    \end{pmatrix} \,, \ 
    A^4  =  \begin{pmatrix}
        1 & 2 & 2 \\
        0 & 1 & 0 \\
        0 & 0 & 1        
    \end{pmatrix} \ \in \ \RNN^{3 \times 3}\,.
    \]
    On the other hand,
    \[
    \chi(A)  =  \begin{pmatrix}
        1 & 1 & 0 \\
        0 & 0 & 1 \\
        0 & 1 & 0
    \end{pmatrix}\,,\ 
    \chi(A^2)  =  \begin{pmatrix}
        1 & 1 & 1 \\
        0 & 1 & 0 \\
        0 & 0 & 1        
    \end{pmatrix} \,, \ 
    \chi(A^3)  =  \begin{pmatrix}
        1 & 1 & 1 \\
        0 & 0 & 1 \\
        0 & 1 & 0        
    \end{pmatrix} \,, \ 
    \chi(A^4)  =  \begin{pmatrix}
        1 & 1 & 1 \\
        0 & 1 & 0 \\
        0 & 0 & 1        
    \end{pmatrix} \ \in\ \B^{3 \times 3}\,,
    \]
    where $\chi(A^k) = \chi(A)^k \in \B^{3 \times 3}$ for all $k \ge 0$.
    For all $i \ge 2$ we have $\chi(A)^{i+2} = \chi(A)^i$.
    The period of~$A$ is $2$, and the index of convergence of~$A$ is also~$2$.
\end{example}

It is further useful to introduce the digraph induced by a nonnegative matrix. Given a nonnegative $n \times n$ matrix $A$, let $G(A) = (V, E)$ be the digraph with the vertex set $V = \{1, 2, \ldots, n\}$ and the edge set~$E$ defined as follows: for $1 \le i, j \le n$, we have $(i, j) \in E$ if and only if the entry $(i, j)$ in $A$ is strictly positive. It is then easy to see that $\chi(A)$ is the adjacency matrix of $G(A)$.

A \emph{$(v_1, v_k)$-path} (or simply a \emph{path}) in a digraph $G$ is an alternating sequence 
\[v_1 \xrightarrow{(v_1, v_2)} v_2 \xrightarrow{(v_2, v_3)} \ldots \xrightarrow{(v_{k - 1}, v_k)} v_k\] of vertices and edges of $G$ with $k \ge 2$. The~\emph{length} of this path is $k - 1$. A path is a \emph{cycle} if $v_1 = v_k$. 
A digraph~$G$ is called \emph{strongly connected} if for every pair $u, v$ of its vertices there are a $(u, v)$-path and a $(v, u)$-path in~$G$.
A maximal (by inclusion) set $S$ of vertices such that for every two vertices $u, v \in S$ there are a $(u, v)$-path and a $(v, u)$-path is called a  \emph{strongly connected component}. Note that in our definition, vertices that do not belong to any cycles are not in any strongly connected component. We transfer the notions of period, index of convergence and primitivity from a matrix~$A$ to its digraph~$G(A)$. Moreover, powering matrices over $\B$ has a natural interpretation in terms of digraphs: for any $k \ge 1$, the entry $(i, j)$ in $\chi(A)^k$ is equal to one if and only if there is an $(i, j)$-path of length~$k$ in $G$.

\begin{example} \label{ex:2}
Let $A \in \RNN^{3 \times 3}$ be the matrix from \Cref{ex:1}.
The digraph $G(A)$ can be visualized as follows.
\begin{center}
\begin{tikzpicture} [node distance = 2cm]
\tikzset{every state/.style={inner sep=1pt,minimum size=1.5em}}

\node [state] at (0, 0) (1) {$1$};
\node [state] at (2, 0) (2) {$2$};
\node [state] at (4, 0) (3) {$3$};

\path [-stealth, thick]
(1) edge [loop left] (1)
(1) edge (2)
(2) edge[bend left=20] (3)
(3) edge[bend left=20] (2)
;
\end{tikzpicture}
\end{center}
Its strongly connected components are $\{1\}$ and $\{2,3\}$.
The top-right entry of $\chi(A^4) = \chi(A)^4$ is the Boolean value~$1$, which indicates that $G(A)$ has at least one path of length~$4$ from $1$ to~$3$. (In fact, there are two such paths in~$G(A)$.)
\end{example}

It is easy to show that a digraph is primitive if and only if it is strongly connected and has period one, see e.g.~\cite[Theorem~3.3]{Minc1988}.
Moreover, if a digraph is strongly connected, then its period is the greatest common divisor of the lengths of all its cycles. If it is not strongly connected, then its period is the least common multiple of the periods of its strongly connected components (we assume that the least common multiple of an empty set is equal to one). Proofs of these statements can be found in \cite{Rosenblatt1957,Kim1982}. In particular, if a digraph is primitive, then there exists $m \ge 1$ such that there is a path of length exactly $m$ between every two vertices. The index of convergence of a primitive digraph is often called its \emph{exponent of primitivity}, or just \emph{exponent}.

\paragraph*{Remarks on computational complexity notions.} We assume that the reader is familiar with basic complexity theory; see, e.g., \cite{Sipser2013}. Since in this paper we discuss the complexity of function problems, the following remark is required.
Slightly abusing notation, we say that a function is computable in \LL (respectively, in \NL) if there exists a deterministic (respectively, nondeterministic) logarithmic-space Turing machine computing this function. Namely, given an input, the machine provides the output on the write-only tape, using only a logarithmic amount of working space in the process. Moreover, our definition of a nondeterministic machine computing a function requires that for each input there is always at least one accepting path, and all accepting paths provide the same output.
Similarly, we speak about \LL- and \NL-algorithms, and say that a function problem is \LL- (respectively, \NL-) complete whenever it is computable in \LL (respectively, in \NL) and the decision problem of checking if the value of the function is at most a given value is \LL-hard (respectively, \NL-hard).
In proofs we will make use of the well-known fact that the composition of two functions computable in \LL (respectively, in \NL) is computable in \LL (respectively, in \NL).
We will also freely use the fact that $\NL = \coNL$ \cite[Section~8.6]{Sipser2013}. In particular, we use the fact that if a property can be decided in \NL, it can be used as a ``subprocedure'' in a nondeterministic logarithmic space machine as follows: guess if the property holds or not, and then verify this fact in~$\NL=\coNL$. 

\paragraph*{Existing results and our contributions.}

There exist several linear-time algorithms for computing the period of a digraph \cite{Balcer1973, Arkin1991,Atallah1982}. For example, the elegant approach of \cite{Balcer1973} for strongly connected digraphs works, roughly speaking, as follows: as long as there is a vertex of outdegree more than one, merge all its successors. This results in a single cycle whose length is equal to the period of the original digraph. This algorithm works in linear space. Somewhat surprisingly, the literature does not seem to have any results concerning the space complexity of computing the period beyond this linear-space upper bound.

Deciding if a digraph is primitive is \NL-hard by a simple reduction from the problem of deciding if a digraph is strongly connected, which is \NL-complete \cite[Problem~8.27]{Sipser2013}: it is enough to add a self-loop edge to every vertex of the input digraph. Then the period of the resulting digraph is one, and hence it is primitive if and only if the original digraph is strongly connected. This leaves open the possibility that computing the period without testing for strong connectivity might be easier. Therefore, our two main results concern the complexity of computing the period for digraphs under restrictions on their strongly connected components: we show that computing the period is \NL-complete even for digraphs with a single non-trivial strongly connected component (\Cref{thm:nlc-one-mat}), but that, in contrast, computing the period of strongly connected digraphs is \LL-complete (\Cref{thm-sc-period}). The latter result
involves a new relation between the period of a strongly connected digraph and reachability properties of undirected graphs (\Cref{lemma-reachability}). We complement these theorems with an \NL-completeness result on the exponent of primitivity (\Cref{thm-index-of-convergence}).

%% file: sec-general.tex
We first show that computing the period of a digraph is \NL-complete, even when it is very close to being strongly connected. A strongly connected component of a digraph is called a \emph{sink} if there is no path from a vertex belonging to it to a vertex not belonging to it.
We call a digraph \emph{almost strongly connected} if it has a unique strongly connected component, and that strongly connected component is a sink.  For example, the digraph in \Cref{fig:st-reachability} (right) is almost strongly connected but not strongly connected, as $s'$ and five other vertices are not in the strongly connected component. In this section we show the following result.

\begin{theorem}\label{thm:nlc-one-mat}
    The period of a digraph is computable in \NL. Deciding if the period of a digraph is one is \NL-complete, even if its period is at most two and it is almost strongly connected.
\end{theorem}

\subsection{\NL-algorithm} 

First we prove the first statement of \Cref{thm:nlc-one-mat}, i.e., that the period is computable in \NL.
Let $G$ be a digraph with $n$ vertices.
As proved in \cite{Schwarz1970} (see also \cite{Jia1987,Heap1966}), the index of convergence of~$G$ is at most $n^2$.
Thus, the period of $G$  is the smallest $p$ such that $A^{n^2} = A^{n^2 + p}$, where $A \in \B^{n \times n}$ is the adjacency matrix of~$G$.

Suppose first that $G$ is strongly connected.
Then its period is at most~$n$, and the above characterisation of the period yields an \NL-algorithm computing it:
check for all $p$ from $1$ to~$n$ whether there exists a pair of vertices $u,v$ such that there is a $(u, v)$-path of length~$n^2$ but not of length $n^2+p$, or vice versa. The period of $G$ is the smallest~$p$ for which no such pair of vertices exists.
This can be done in~\NL by guessing $(u,v)$-paths of the required lengths in~$G$.
We will improve this to an \LL-algorithm in the next section.

For a general (not necessarily strongly connected) digraph~$G$, the period equals the least common multiple of the periods of the strongly connected components \cite{Rosenblatt1957}.
Therefore, it can be exponential in $n$, and thus has to be represented in binary.
The list of the strongly connected components of~$G$ can be computed in~\NL (without storing it on the work tape).
Since \NL-computable functions can be composed, by using the algorithm from the previous paragraph, we can also compute in~\NL the list of their periods, say $n_1, \ldots, n_k$, in their unary representation.
Using the compositionality of \NL-computable functions again,
to show that the overall period can be computed in~\NL, it suffices to show that the \emph{lcm problem} can be computed in~\NL, where the lcm problem is to compute, in binary representation, the least common multiple of a given list of numbers represented in unary.
Hence, our remaining task is to show that the lcm problem can be computed in~\NL.
In fact, we show that it can be computed in~\LL.
The fact that the input to the lcm problem is in unary simplifies our task.

Denote by $n_1, \ldots, n_k$ the input (in unary) to the lcm problem and by $\ell$ the output, i.e., the least common multiple of $n_1, \ldots, n_k$.
We proceed in two steps, again using the compositionality of functions computable in \LL.
In the first step, we compute in~\LL the (unique) sorted list of prime numbers $p_1, \ldots, p_k$ (in binary) 
and positive integers $m_1, \ldots, m_k$ (in unary) with $p_1 < p_2 < \ldots < p_k$ and $\ell = \prod_{j=1}^k p_j^{m_j}$.
This can be done naively.
Specifically, we count up a binary counter~$c$ from one to the largest value among
the numbers $n_i$, and for every value of $c$ we perform the following process. We first check if $c$ is prime, by dividing $c$ by the numbers from $2$ to $c-1$. If it is, we check if $c$ divides $n_i$ for some~$i$. If the answer to either of these two checks is no, we proceed to the next value of $c$.
If both answers are yes, the current value of $c$ is the next prime $p_j$ from the list, and we compute its multiplicity~$m_j$ by incrementing $m_j$ as long as $p_j^{m_j}$ divides at least one of the numbers~$n_i$. 
Since we only need a constant number of numbers on the work tape at the same time and these numbers are encoded in binary, the space needed for the described arithmetic operations is logarithmic in $n_1 + \cdots + n_k$, i.e., logarithmic in the input.

In the second and final step, we need to compute a binary representation of $\ell = \prod_{j=1}^k p_j^{m_j}$, where the numbers $p_j, m_j$ are given in binary and unary, respectively.
To do that, we make a new list consisting, for each $j$, of  $m_j$ copies of~$p_j$. We then compute the binary representation of the product of all numbers from the list: this is known as iterated multiplication and is computable in \LL \cite[Theorems 1.1, 1.2, 1.5]{ChiuDL01}.

\subsection{\NL-hardness}

We now show \NL-hardness stated in \Cref{thm:nlc-one-mat}. Given a digraph $G = (V, E)$ and two its vertices $s$ and~$t$, the $(s,t)$-reachability problem asks if there exists an $(s, t)$-path in $G$.
The $(s,t)$-reachability problem is well known to be \NL-complete \cite[Theorem~8.25]{Sipser2013}.
We will make some assumptions on the digraph, which are justified by the following folklore lemma.

\begin{lemma} \label{lem:st-reach-ass}
The $(s,t)$-reachability problem in digraphs is \NL-hard even under the following assumptions:
the digraph is acyclic, $s$ has no incoming edges, $t$ has no outgoing edges, there is only one vertex $s' \ne s$ without incoming edges,  and there is only one vertex~$t' \ne t$ without outgoing edges.
\end{lemma}
\begin{proof}
We provide an \ACZ reduction from the $(s,t)$-reachability problem to the restricted version from the statement.
We make the following modifications in turn, none of which affects the existence of a path from $s$ to~$t$.

The digraph can be made acyclic using the following standard transformation: create a new digraph whose state set consists of $|V|$ copies of $V$. For each edge $(u, v) \in E$ and each $1 \le i \le |V|-1$, create an edge going from the $i$th copy of $u$ to the $(i + 1)$st copy of $v$ in the new digraph, and also an edge going from the $i$th copy of $u$ to the $(i + 1)$st copy of $u$.
Since a shortest $(s, t)$-path in $G$, if it exists, has length at most $|V| - 1$, we get that it exists if and only if there is a path from the first copy of $s$ to the $|V|$th copy of $t$ in the new digraph.
For simplicity, call the $1$st copy of~$s$ and the $|V|$th copy of~$t$ again $s$ and $t$, respectively.

By construction, there are now no edges incoming to~$s$ and no edges outgoing from~$t$.
Finally, we merge vertices in such a way that there is only one vertex $s' \ne s$ without incoming edges and only one vertex $t' \ne t$ without outgoing edges.
\end{proof}

We now proceed to the proof of \NL-hardness stated in \Cref{thm:nlc-one-mat}.
    We construct a digraph $G' = (V', E')$ by modifying $G$ as follows. First, subdivide every edge of~$G$ into a directed path of length two by adding a new vertex for every edge. Denote by $V'$ the set of vertices in the obtained digraph, and by~$E'$ the set of thus obtained edges together with the edges $(s, t), (t', t), (t, s)$. See \Cref{fig:st-reachability} for an~example.

\begin{figure}[ht]\centering
\begin{subfigure}[c]{0.45\textwidth}
\begin{tikzpicture} [node distance = 2cm]
\tikzset{every state/.style={inner sep=1pt,minimum size=1.5em}}

\node [state] at (0, 0) (s) {$s$};
\node [state] at (2, 0) (p2) {};

\node [state] at (4, 0) (p4) {$t'$};
\node [state] at (2, -2) (q2) {$s'$};
\node [state] at (4, -2) (r2) {};

\node [state] at (6, 0) (t) {$t$};
\node [state] at (4, 2) (z1) {};

\path [-stealth, thick]

(s) edge [] node[above] {} (p2)
(p2) edge [] node[above] {} (p4)

(q2) edge [] node[above] {} (p2)
(q2) edge [] node[above] {} (r2)
(r2) edge [] node[above] {} (t)
(r2) edge [] node[above] {} (p4)

(p2) edge [] node[above] {} (z1)
(z1) edge [] node[above] {} (t)

;
\end{tikzpicture}
\end{subfigure}
\hspace{0.2cm}
\begin{subfigure}[c]{0.45\textwidth}
\begin{tikzpicture} [node distance = 2cm]
\tikzset{every state/.style={inner sep=1pt,minimum size=1.5em}}

\clip  (-0.3, -3.15) rectangle (6.3,3.15);

\node [state] at (0, 0) (s) {$s$};
\node [state, dotted] at (1, 0) (p1) {};
\node [state] at (2, 0) (p2) {};
\node [state, dotted] at (3, 0) (p3) {};
\node [state] at (4, 0) (p4) {$t'$};
\node [state, dotted] at (2, -1) (q1) {};
\node [state] at (2, -2) (q2) {$s'$};
\node [state, dotted] at (3, -2) (r1) {};
\node [state] at (4, -2) (r2) {};
\node [state, dotted] at (5, -1) (g1) {};
\node [state, dotted] at (4, -1) (k1) {};

\node [state] at (6, 0) (t) {$t$};
\node [state, dotted] at (3, 1) (z0) {};
\node [state] at (4, 2) (z1) {};
\node [state, dotted] at (5, 1) (z2) {};

\path [-stealth, thick]

(s) edge [] node[above] {} (p1)
(p1) edge [] node[above] {} (p2)
(p2) edge [] node[above] {} (p3)
(p3) edge [] node[above] {} (p4)

(q1) edge [] node[above] {} (p2)
(q2) edge [] node[above] {} (q1)
(q2) edge [] node[above] {} (r1)
(r1) edge [] node[above] {} (r2)
(r2) edge [] node[above] {} (g1)
(r2) edge [] node[above] {} (k1)

(g1) edge [] node[above] {} (t)
(k1) edge [] node[above] {} (p4)

(p2) edge [] node[above] {} (z0)
(z0) edge [] node[above] {} (z1)
(z1) edge [] node[above] {} (z2)
(z2) edge [] node[above] {} (t)

(p4) edge [] node[above] {} (t)

;

\draw [-stealth, thick]
(t) ..  controls  ($(t) + (-1, 4)$) and ($(s) + (1, 4)$)  .. node[above] {} (s);
\draw [-stealth, thick]
(s) ..  controls  ($(s) + (1, -4)$) and ($(t) + (-1, -4)$)  .. node[above] {} (t);
\end{tikzpicture}
\end{subfigure}

\caption{The original digraph $G$ (left) and the digraph $G'$ obtained in the reduction in the proof of \Cref{thm:nlc-one-mat} (right). Dotted vertices are obtained by subdividing each edge into two.}\label{fig:st-reachability}
\end{figure}
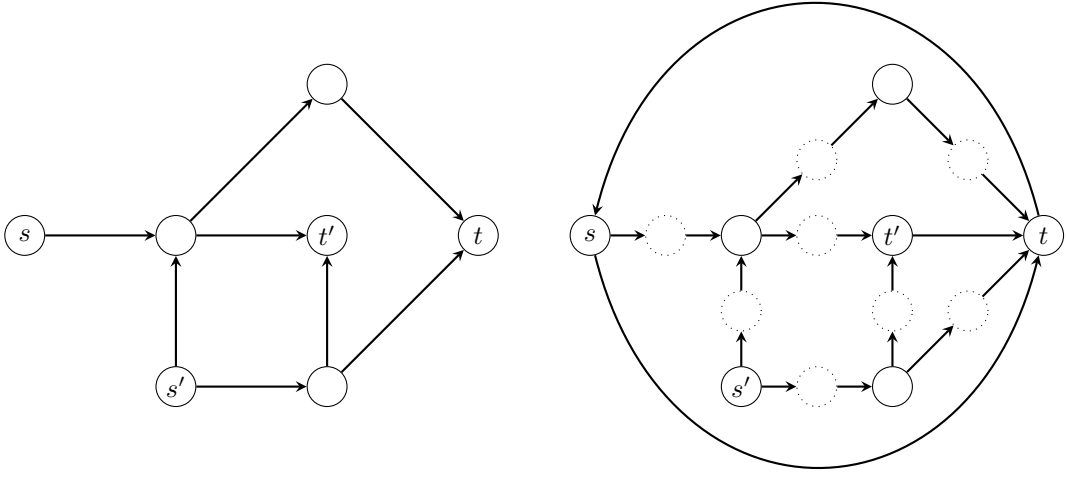

    Observe that by construction the digraph $(V', E' \setminus \{(t, s)\})$ is acyclic. Therefore, since $G'$ differs from it by only one edge $(t, s)$, every cycle in $G'$ must contain $(t, s)$. Furthermore, $(V', E' \setminus \{(t, s)\})$ contains an $(s, t)$-path of even length if and only if there is an $(s, t)$-path in $G$. Hence, there is a cycle of odd length in $G'$ if and only if there is an $(s, t)$-path in $G$. Thus, since $G'$ contains a cycle consisting of two edges $(s, t)$ and $(t, s)$, the period of $G'$ is one if and only if there is an $(s, t)$-path in $G$, otherwise it is two.

    To conclude the proof, it remains to note that by construction every cycle in $G'$ is contained in the unique sink strongly connected component (which contains $s$ and $t$), so $G'$ is almost strongly connected.

%% file: sec-sc.tex
We now show that if a digraph is strongly connected, computing the period becomes easier, namely \LL-complete. The proof of the following theorem is provided in the next two subsections.

\begin{theorem}\label{thm-sc-period}
    The period of a strongly connected digraph is computable in \LL. Deciding if a strongly connected digraph has period one is \LL-complete, even for digraphs of period at most two.
\end{theorem}

\subsection{\LL-algorithm} 

We first show how to compute the period in \LL. We will need the following classical result. Let $G = (V, E)$ be a digraph, and let $S_1, \ldots, S_t$ be a partition of its state set $V$ into~$t$ pairwise disjoint subsets. 
We call such partition \emph{$t$-consistent} if for every edge $(u, v) \in E$, it holds true that if $u \in S_i$ then $v \in S_{i + 1~(\modu t)}$.

\begin{proposition}[see e.g.\ Theorem 10.5.1 in \cite{BangJensen2008}]\label{prop-partition}
    Let $p$ be the period of a strongly connected digraph $G = (V, E)$. Then there exists a $p$-consistent partition of its vertex set $V$.
\end{proposition}

Let $G = (V, E)$ be a strongly connected digraph. Given a natural number $k \ge 1$, consider the undirected graph $G'_k = (V', E')$ defined as follows. We take $V' = V \times \ZZ_k$, where $\ZZ_k$ is the set of residue classes modulo $k$. There is an edge between vertices $(u, i)$ and $(v, j)$ in $G'_k$ if and only if $(u, v) \in G$ and $j = i + 1~(\modu k)$. We are going to show the following relation between reachability properties of $G'_k$ and the period of $G$.

\begin{lemma}\label{lemma-reachability}
    Let $v$ be a vertex of $G$ and $k$ be an integer number between $1$ and $|V|$. We have that~$k$ divides the period of $G$ if and only if there is no undirected path in $G'_k$ between $(v, 0)$ and $(v, i)$ for any~$i \ne 0$.
\end{lemma}
\begin{proof}
It follows from \Cref{prop-partition} that a $k$-consistent partition of $V$ exists if and only if $k$ divides the period of $G$. Indeed, if such a partition exists, the length of every cycle of $G$ is clearly divisible by $k$, and thus~$k$ divides the period of $G$.
Conversely, if $k$ divides the period~$p$ and $S_0, \ldots, S_{p-1}$ is a $p$-consistent partition of $V$, then $S'_0, \ldots, S'_{k-1}$ with $S'_i = \bigcup_{j=0}^{\frac{p}{k}-1} S_{i+j k}$ is a $k$-consistent partition of~$V$.
It remains to show that a $k$-consistent partition exists if and only if $G'_k$ has no path between $(v, 0)$ and $(v, i)$ for any $i \ne 0$.

Concerning the forward implication, observe that every undirected path in $G'_k$ between vertices $(v, 0)$ and $(u, j)$ such that $u$ and $v$ are in the same set $S_h$ must have $j = 0$.
Indeed, by construction of~$G'_k$, going forward along an edge of $G$ adds one to the second component of the corresponding vertex of $G'_k$, and going backwards subtracts one. 
Thus, the existence of a $k$-consistent partition implies that for every vertex $v \in V$, there is no path in $G'_k$ between $(v, 0)$ and $(v, i)$ for any $i \ne 0$.

Conversely, suppose that there is no path in $G'_k$ between $(v, 0)$ and any $(v,i)$ with $i \ne 0$.
Then for any path in $G'_k$ between $(v, \ell)$ and $(v, \ell')$ we have $\ell = \ell'$, as any such path can be translated to a path between $(v, 0)$ and $(v, \ell'-\ell)$ by construction of $G'_k$.
It follows that if there is a path in $G'_k$ between $(v,0)$ and some vertex $(u,j)$, then there is no path between $(u,j)$ and any $(u,j')$ with $j' \ne j$: such a path could be extended to a path between $(v, \ell)$ and $(v, \ell+j'-j)$ for some~$\ell$.
Thus, the partition $S_0, \ldots, S_{k-1}$ of $V$, where $S_j$ contains those vertices $u$ such that $(u,j)$ is reachable from $(v,0)$ is $k$-consistent.
\end{proof}
 
To compute the period of $G$, fix an arbitrary vertex $v$ in $G$ and check the property of \Cref{lemma-reachability} for all values of $k$ from $1$ to $|V|$. The largest value $k$ such that there is no path in~$G'_k$ between $(v, 0)$ and $(v, i)$ for any $i \ne 0$ is the period of $G$. Since the existence of a path between two vertices in an undirected graph can be checked in \LL~\cite{Reingold2008}, and for any given $k$ the graph $G'_k$ can be constructed in~\LL, the described algorithm also works in \LL.

\subsection{\LL-hardness} 

It remains to prove \LL-hardness, which follows from the following lemma.

\begin{lemma}\label{lemma-sc-period-hard}
    Let $G$ be a  strongly connected digraph of period at most two. Deciding if the period of~$G$ is one is \LL-hard.
\end{lemma}

\begin{proof}
    We provide an \ACZ reduction from the \textsc{ORD} problem, which is \LL-hard under an \ACZ reduction \cite{Etessami1997}  (see \cite{Immerman1995} for a hierarchy of reductions). The input of the \textsc{ORD} problem is a directed path $G = (V, E)$ and two vertices $s, t \in V$ on this path. The path is defined by its successor relation $S(x, y)$ such that $S(x, y)$ is true if and only if $(x, y) \in E$.  Intuitively, this means that $G$ is provided as an unsorted list of edges and is promised to consist of a single directed path. The \textsc{ORD} problem then asks if $s$ precedes $t$ in the unique total ordering consistent with $S$.

    The reduction goes as follows. Let $v$ be the first vertex on the directed path $G$, and $v'$ be the last vertex on this path. We can assume that $v'$ is different from both $s$ and $t$. We substitute every edge of~$G$ with a path of length two by introducing new vertices, and we also add two additional new vertices $s'$ and $v'$. We now remove the edge $e$ going to $s$, and add an edge from the beginning of $e$ to~$s'$, and we also add the edge $(s', s)$. We also add the edge $(v', v)$. Finally, we add three edges $(v, s)$, $(v, t)$ and $(v, v')$.  Denote the thus obtained digraph as $G' = (V', E')$. The construction is illustrated~in~\Cref{fig:period-l-hard}.

    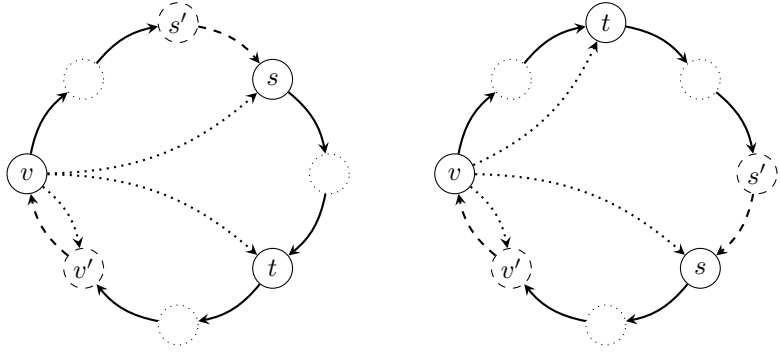
\begin{figure}[ht]\centering
{\begin{tikzpicture} [node distance = 2cm]
\tikzset{every state/.style={inner sep=1pt,minimum size=1.5em}}

\node [state] at (0, 0) (1) {$v$};
\node [state, dotted] at (0.75, 1.25) (2) {};
\node [state, dashed] at (2, 2) (3) {$s'$};
\node [state] at (3.25, 1.25) (4) {$s$};
\node [state, dotted] at (4, 0) (5) {};
\node [state] at (3.25, -1.25) (6) {$t$};
\node [state, dotted] at (2, -2) (7) {};
\node [state, dashed] at (0.75, -1.25) (8) {$v'$};

\path [-stealth, thick]
(1) edge [bend left=20] node[above] {} (2)
(2) edge [bend left=20] node[above] {} (3)
(3) edge [bend left=20, dashed] node[above] {} (4)
(4) edge [bend left=20] node[above] {} (5)
(5) edge [bend left=20] node[above] {} (6)
(6) edge [bend left=20] node[above] {} (7)
(7) edge [bend left=20] node[above] {} (8)
(8) edge [bend left=20, dashed] node[above] {} (1)

(1) edge [bend left=20, dotted] node[above] {} (8)

(1) edge [bend right=20, dotted] node[above] {} (4)

(1) edge [bend left=20, dotted] node[above] {} (6)

;

\end{tikzpicture}}\hspace{1cm}
{\begin{tikzpicture} [node distance = 2cm]
\tikzset{every state/.style={inner sep=1pt,minimum size=1.5em}}

\node [state] at (0, 0) (1) {$v$};
\node [state, dotted] at (0.75, 1.25) (2) {};
\node [state] at (2, 2) (3) {$t$};
\node [state, dotted] at (3.25, 1.25) (4) {};
\node [state, dashed] at (4, 0) (5) {$s'$};
\node [state] at (3.25, -1.25) (6) {$s$};
\node [state, dotted] at (2, -2) (7) {};
\node [state, dashed] at (0.75, -1.25) (8) {$v'$};

\path [-stealth, thick]

(1) edge [bend left=20] node[above] {} (2)
(2) edge [bend left=20] node[above] {} (3)
(3) edge [bend left=20] node[above] {} (4)
(4) edge [bend left=20] node[above] {} (5)
(5) edge [bend left=20, dashed] node[above] {} (6)
(6) edge [bend left=20] node[above] {} (7)
(7) edge [bend left=20] node[above] {} (8)
(8) edge [bend left=20, dashed] node[above] {} (1)

(1) edge [bend left=20, dotted] node[above] {} (8)

(1) edge [bend right=20, dotted] node[above] {} (3)
(1) edge [bend left=20, dotted] node[above] {} (6)

;
\end{tikzpicture}}
\caption{Digraphs obtained from the directed paths $v \to s \to t \to v'$ (left) and $v \to t \to s \to v'$ (right). Dotted vertices represent the result of subdividing edges, dashed vertices are the additional new vertex $s'$ and the end of the path $v'$. The three edges added at the end are dotted, and dashed edges are the ``special'' edges going from $s'$ and~$v'$. Depending on the order of $s$ and $t$, these ``special'' edges are included in different cycles and thus change the parity of the period.}\label{fig:period-l-hard}
\end{figure}

We now argue that the period of $G'$ is one if and only if $t$ comes before $s$ in $G$, otherwise the period of $G'$ is two.
It is easy to see that $G'$ has exactly four cycles, and all of them contain vertex $v$. Hence, these four cycles correspond to the four edges outgoing from $v$. Moreover, the period of $G'$ is at most two since it contains a cycle of length two consisting of the vertices $v$ and $v'$. The cycle containing all vertices of $G'$ always has even length, since it consists of pairs of edges obtained from subdividing the edges of $G$ together with two additional edges $(s', s)$ and $(v', v)$. Similarly, the cycle containing the edge $(v, s)$ always has even length, since it consists of subdivided edges together with $(v, s)$ and~$(v', v)$. The only difference is the parity of the length of the cycle containing the edge $(v, t)$: if $t$ comes before~$s$ in $G$, then this cycle consists of subdivided edges together with $(v, t)$, $(s', s)$ and $(v', v)$, and hence its length is odd. If $t$ comes after $s$, then this cycle does not contain $(s', s)$, and by a similar argument its length is even.
\end{proof}

We remark that, intuitively, the main reason why the complexity goes from \LL-complete for strongly connected digraphs to \NL-complete for almost strongly connected digraphs is that \Cref{prop-partition} is no longer true for almost strongly connected digraphs, making this case general enough to express reachability in digraphs.
Indeed, the digraph with vertices $1, 2, 3$ and edges $(1,2), (1,3), (2,3), (3,2)$ has period~$2$ and no $2$-consistent partition.

%% file: sec-exponent.tex
We conclude by studying the complexity of computing the index of convergence of a digraph. We show that it is \NL-complete by proving the following result.

\begin{theorem} \label{thm-index-of-convergence}
    The index of convergence of a digraph is computable in \NL. The problem of deciding if the exponent (that is, the index of convergence) of a primitive digraph with $n$ vertices is at most~$\frac{n}{2}$ is \NL-complete. 
\end{theorem}

\subsection{\NL-algorithm}

Let $G = (V, E)$ be a digraph.
Denote by~$P$ its period, that is, the least common multiple of the periods of the strongly connected components.
We want to show the first statement of \Cref{thm-index-of-convergence}, i.e., that the index of convergence of~$G$ can be computed in~\NL. As proved in \cite{Schwarz1970} (see also \cite{Jia1987,Heap1966}), the index of convergence is at most $|V|^2$. Let $0 \le K < |V|^2$. It  suffices to show that one can determine in~\NL if  the index of convergence is greater than $K$.
This condition is satisfied if and only if there are vertices $s, t$ such that either
\begin{enumerate}[(i)]
    \item there is no
$(s, t)$-path of length~$K$ but for all large enough~$i$ there is an $(s, t)$-path of length $K + i \cdot P$; or
\item there is an $(s, t)$-path of length~$K$ but for infinitely many (in fact, almost all)~$i$ there is no $(s, t)$-path of length $K + i \cdot P$.
\end{enumerate}

To check this in~\NL, we non-deterministically guess vertices $s, t$.
The existence of an $(s, t)$-path of length~$K$ can easily be determined in~\NL. It remains to show how to determine in~\NL whether for all large enough~$i$ there exists an $(s, t)$-path of length $K + i \cdot P$. To do that, we use the following lemma.

\begin{lemma}\label{lemma-exponent-in-NL}
    Let $s, t$ be two vertices, and let $K \ge 0$ be an integer.
    The following statements are equivalent.
    \begin{enumerate}[(1)]
        \item For all large enough integers $i \ge 0$ there exists an $(s, t)$-path of length $K + i \cdot P$.
        
        \item There exists a strongly connected component~$C$ and an $(s, t)$-path through~$C$ of length $\ell$ with  $\ell = K~(\modu p)$,
where $p$ is the period of~$C$.
    \end{enumerate}
\end{lemma}

\begin{proof}
``(1) $\Longrightarrow$ (2)''.
For large enough~$i$, any $(s, t)$-path of length $K + i \cdot P$ passes through some strongly connected component $C$ of some period $p$. If the length of such path is $\ell$ and $\ell = K~(\modu P)$, then $\ell = K~(\modu p)$ since $p$ divides $P$.

``(2) $\Longrightarrow$ (1)''.
Let $c \in C$ be a vertex on an $(s, t)$-path~$w$ of length $\ell$ with  $\ell = K~(\modu p)$, where~$p$ is the period of~$C$.
From the definition of~$p$, for all large enough~$i$ there is a cycle of length $i \cdot p$ containing~$c$. 
Conjoining~$w$ with these cycles yields $(s, t)$-paths of lengths $K + i \cdot p$ for all large enough~$i$.
Since $p$ divides $P$, statement~(1) follows.
\end{proof}

Now, we first non-deterministically guess a strongly connected component~$C$ and compute its period~$p$ in~\LL (using \Cref{thm-sc-period}). Using \Cref{lemma-exponent-in-NL}, we then check in~\NL the existence of an $(s, t)$-path through~$C$ of length $\ell$ with  $\ell = K~(\modu p)$: if such a path exists, there must be one of length at most~$2 |V| p$.
Thus, we get that one can compute the index of convergence in~\NL.

\subsection{\NL-hardness}\label{subsec-exp-hardness}

To prove \NL-hardness, we reduce from the $(s, t)$-reachability problem in digraphs with assumptions as in \Cref{lem:st-reach-ass} (including those on the special vertices $s',t'$). Let $G = (V, E)$ be such a digraph. 
Construct a digraph $G' = (V', E')$ by adding vertices and edges to $G$ as follows. Take $V' = V \cup U$, where $U = \{u_1, \ldots, u_{|V|}\}$ is a set of $|V|$ fresh vertices. Thus, $|V'| = 2|V|$. Take $E'$ to be the set $E$ together with an edge from every vertex in $V'$ to $s$, an edge from $t$ to every vertex in $V'$, and the edges $(t',u_1), (u_i, u_{i+1})$ for all $1 \le i \le |V|-1$, and $(u_{|V|},t)$, which form a $(t',t)$-path via~$U$. This construction is illustrated in \Cref{fig:exponent}.

\begin{figure}[h]\centering
\begin{subfigure}[T]{0.30\textwidth}{\begin{tikzpicture} [node distance = 2cm]
\tikzset{every state/.style={inner sep=1pt,minimum size=1.5em}}

\node [state] at (2, 0) (s) {$s$};
\node [state] at (4, 0) (tp) {$t'$};
\node [state] at (6, 0) (t) {$t$};
\node [state] at (4, 2) (sp) {$s'$};

\path [-stealth, thick]

(s) edge [] node[above] {} (tp)

(sp) edge [] node[above] {} (tp)
(sp) edge [] node[above] {} (t)

;
\end{tikzpicture}}
\end{subfigure}
\hspace{0.9cm} 
\begin{subfigure}[T]{0.30\textwidth}
\begin{tikzpicture} [node distance = 2cm]
\tikzset{every state/.style={inner sep=1pt,minimum size=1.5em}}

\node [state] at (2, 0) (s) {$s$};
\node [state] at (4, 0) (tp) {$t'$};
\node [state] at (6, 0) (t) {$t$};
\node [state] at (4, 2) (sp) {$s'$};

\node [state] at (2, -2) (u1) {$u_1$};
\node [] at (4, -2) (u2) {$...$};
\node [state] at (6, -2) (u3) {$u_4$};

\path [-stealth, thick]

(s) edge [] node[above] {} (tp)

(sp) edge [] node[above] {} (tp)
(sp) edge [] node[above] {} (t)

(tp) edge [] node[above] {} (u1)
(u1) edge [] node[above] {} (u2)
(u2) edge [] node[above] {} (u3)
(u3) edge [] node[above] {} (t)

(sp) edge [dotted] node[above] {} (s)
(tp) edge [dotted, bend right=20] node[above] {} (s)
(t) edge [dotted, bend right=30] node[above] {} (s)
(u1) edge [dotted] node[above] {} (s)
(u2) edge [dotted] node[above] {} (s)
(u3) edge [dotted] node[above] {} (s)

(t) edge [dashed] node[above] {} (tp)
(t) edge [dashed] node[above] {} (u1)
(t) edge [dashed] node[above] {} (u2)
(t) edge [dashed, bend left=20] node[above] {} (u3)
(t) edge [dashed, bend right=20] node[above] {} (sp)

(s) edge [dotted, loop left] node[above] {} (s)
(t) edge [dashed, loop right] node[above] {} (t)

;
\end{tikzpicture}
\end{subfigure}
\caption{The original digraph $G$ (left) and the digraph $G'$ obtained in the reduction in \Cref{subsec-exp-hardness} (right).}\label{fig:exponent}
\end{figure}
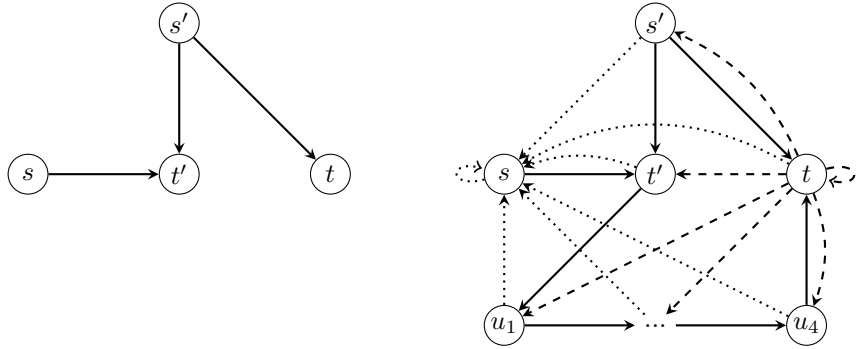

We start by showing that $G'$ is primitive. First, observe that $G'$ is strongly connected. Indeed, there is an edge in $G'$ from every vertex to $s$, there is a path in the original digraph $G$ (and thus in~$G'$) from every vertex to either $t$ or $t'$ by our assumptions, there is a path in $G'$ from $t'$ to $t$, and there is an edge in $G'$ from $t$ to every vertex of $G'$. Moreover, the period of $G'$ is one since it contains the edge~$(t, t)$. Hence, $G'$ is primitive.

Now, assume that there is an $(s, t)$-path in $G$. Then there is an $(s, t)$-path of length at most $|V| - 2$, since $t'$ does not have any outgoing edges in $G$, and hence this path cannot contain $t'$. Since in $G'$ there is an edge from every vertex to $s$, an edge $(t, t)$, and an edge from $t$ to every vertex, there is thus a path of length exactly $|V| = \frac{1}{2}|V'|$ between every two vertices of $G'$. Hence, the exponent of $G'$ is at most $\frac{1}{2}|V'|$.

Conversely, assume now that there is no $(s, t)$-path in $G$.
Then every $(s, t)$-path in~$G'$ must go through the $(t', t)$-path of length $|V|+1$ containing all vertices from $U$. Hence, the length of every $(s, t)$-path is at least $\frac{1}{2}|V'| + 1$, so the exponent of $G'$ is also at least $\frac{1}{2}|V'| + 1$.